\documentclass[letterpaper, 10pt, twocolumn]{article}
\usepackage{titlesec}
\titleformat{\section}{\normalfont\large\bfseries}{\thesection}{1em}{} 
\titleformat{\subsection}{\normalfont\normalsize\bfseries}{\thesubsection}{1em}{} 
\titleformat{\subsubsection}{\normalfont\small\bfseries}{\thesubsubsection}{1em}{} 
\usepackage[letterpaper,top=0.75in,bottom=1in,left=0.625in,right=0.625in]{geometry}
\usepackage{cite}
\usepackage{amsmath,amssymb,amsfonts,amsthm}
\usepackage{algorithm,algpseudocode}
\usepackage{graphicx}
\usepackage{mathtools}
\usepackage{balance}
\usepackage{comment}

\usepackage{arydshln}

\usepackage[final]{hyperref} 
\hypersetup{
colorlinks=true, 
linkcolor=black, 
citecolor=black, 
filecolor=magenta, 
urlcolor=black         
}

\newcommand{\R}{\mathbb{R}}
\newcommand{\I}{\mathbb{I}}

\newcommand{\norm}[1]{\left\lVert#1\right\rVert}
\newcommand{\fe}{\mathrm{f}}
\newcommand{\ve}{\mathrm{v}}
\newcommand{\ee}{\mathrm{e}}
\newcommand{\Ve}{\mathrm{V}}
\newcommand{\Ze}{\mathrm{Z}}

\newcommand{\QED}{\hfill$\blacksquare$}

\newcommand\scalemath[2]{\scalebox{#1}{\mbox{\ensuremath{\displaystyle #2}}}}

\newtheorem{theorem}{Theorem}
\newtheorem{corollary}{Corollary}
\newtheorem{remark}{Remark}

\newtheorem{proposition}{Proposition}
\newtheorem{assumption}{Assumption}
\newtheorem{definition}{Definition}

\newcommand{\email}[1]{\href{mailto:#1}{#1}}

\title{\LARGE \bf Configuration-Constrained Tube MPC for Periodic Operation}

\author{Filippo Badalamenti, Jose A. Borja-Conde, Sampath Kumar Mulagaleti, Boris Houska,\\ Alberto Bemporad, Mario Eduardo Villanueva\thanks{F.~Badalamenti (Corresponding Author, \email{filippo.badalamenti@imtlucca.it}), S.K.~Mulagaleti, M.E.~Villanueva and A.~Bemporad are with IMT School for Advanced Studies Lucca, Piazza San Francesco 19, 55100, Italy. J.A.~Borja-Conde is with the Department of Engineering,
Universidad Loyola Andalucía, Spain. B.~Houska is with ShanghaiTech University, China.}}

\begin{document}

\date{}
\maketitle
\thispagestyle{empty}
\pagestyle{empty}

\begin{abstract}
Periodic operation often emerges as the economically optimal mode in industrial processes subject to varying economic or environmental conditions. This paper proposes a robust model predictive control (MPC) framework for uncertain systems described by polytopic linear difference inclusions (LDIs), whose dynamics evolve as convex combinations of finitely many linear control systems subject to additive disturbances. The robust control problem is formulated as a convex optimization program over configuration-constrained polytopic tubes, with the objective of tracking a periodic tube that is optimal for a given economic criterion. Artificial variables embedded in the formulation ensure recursive feasibility and robust constraint satisfaction under changes in the economic criterion, and convergence to the optimal periodic tube when the criterion remains constant. To improve computational efficiency, we introduce a quadratic upper approximation of the periodic cost under a Lipschitz continuity assumption, yielding a Quadratic Program (QP) that preserves the above theoretical guarantees. The approach is demonstrated on an illustrative example and an eight-state ball--plate system.
\end{abstract}

\section{Introduction} \label{sec:Sect1}
Model Predictive Control (MPC) is widely adopted for constrained multivariable systems thanks to its ability to optimize performance online while enforcing state and input constraints~\cite{MAYNE2000}. In many processes, periodic operation often emerges as the economically optimal mode~\cite{lee2001model}, with the notable exception of linear time-invariant (LTI) systems with convex stage costs, where the optimal operation reduces to a unique steady state~\cite{HOUSKA201779}. Such periodicity typically arises from exogenous factors—e.g., cyclic demands, time-varying prices—or from the nonlinear nature of underlying dynamics~\cite{bittanti2009periodic}.

Economic MPC (E‑MPC)~\cite{ellis2017economic} formalizes this by embedding an economic stage cost that optimizes transient behavior and asymptotic performance. Under strict dissipativity, E‑MPC guarantees convergence to an economically optimal equilibrium in the time‑invariant case~\cite{angeli2011average}, and extensions to periodic operation exist for nonlinear systems~\cite{zanon_7549077}. Furthermore, including an artificial periodic trajectory as an online decision variable preserves recursive feasibility when the optimal periodic orbit changes, and --under suitable constraints on this trajectory and reference-generic terminal ingredients-- ensures closed‑loop performance and stability~\cite{KOHLER2020185}. Periodic E-MPC formulations without terminal constraints have also been proposed, although their closed-loop guarantees require suitable dissipativity conditions and may rely on multi-step MPC updates~\cite{MULLER2016128}.

An alternative is to track a periodic solution that is optimal for a given economic criterion~\cite{wurth2011two,LIMON_Periodic_7172461}. Two-layer architectures~\cite{wurth2011two} compute this solution via an upper-level dynamic real-time optimizer (DRTO) and track it with a lower-level MPC~\cite{krupa2020implementation}; such schemes achieve local equivalence with E-MPC~\cite{zanon2017periodic}, but their performance may degrade under abrupt parameter changes or model mismatch, potentially compromising stability. Single-layer formulations integrate these tasks into one optimization problem~\cite{zanin2002integrating,LIMON2012490}, ensuring recursive feasibility and convergence to the best admissible periodic operation~\cite{singleLayerPeriodicMPC}. However, these formulations introduce new challenges: online updates of the economic criterion increase problem dimensionality, and the associated cost is often non-quadratic, which precludes a Quadratic Program (QP) structure for linear systems.

Uncertainty introduces additional challenges in periodic MPC design~\cite{zanon_7549077}. Disturbances, unmodeled dynamics, and parameter variations prevent the system from following a single periodic trajectory; instead, one can only guarantee that the state remains within a bounded region. Consequently, the control objective shifts from tracking a trajectory to tracking a \emph{periodic sequence of sets} that is robustly forward invariant under all admissible uncertainties. Tube-based MPC (TMPC) addresses this by optimizing over Robust Forward Invariant Tubes (RFITs) that contain all possible trajectories under disturbances. Among polyhedral tube parameterizations --see~\cite{HOUSKA2025100992} for a recent survey-- rigid~\cite{MAYNE_RakovicRigidTMPC2005219,Implicit_rakovic2023}, homothetic~\cite{RAKOVIC_HomoTMPC20121631,Langson2004}, and elastic~\cite{ETMPC_7525471,Fleming2015} tubes have been widely studied, each trading off flexibility and computational complexity. More recently, configuration-constrained (CC) polytopic tubes~\cite{CCTMPC_VILLANUEVA2024111543} have emerged as a flexible alternative, enabling vertex-based invariance checks that are both necessary and sufficient. Attempts to incorporate periodicity under uncertainty have relied on rigid~\cite{wang2019robust}, homothetic~\cite{dong2020homothetic}, or ellipsoidal~\cite{VILLANUEVA_ellipsoids2017311} tubes, but these approaches either restrict the tube geometry or incur a comparatively high online complexity.

\subsubsection*{Contribution} 
This work extends configuration-constrained tube MPC (CCTMPC) from tracking piecewise-constant references~\cite{CCTMPC_Tracking_10546969} to robust tracking of \emph{periodic tubes} that are optimal for a given economic criterion. The uncertain system is modeled as a polytopic LDI, which captures both additive disturbances and polytopic parametric uncertainty in a unified framework. The same representation may also provide a polytopic outer approximation of nonlinear or state-dependent dynamics over a compact operating region. This uncertainty structure is directly handled by configuration-constrained tubes, offering greater geometric flexibility than rigid or homothetic parameterizations. Stability and convergence for fixed economic criteria are established through linear admissibility and contractivity conditions relative to the artificial PFIT, avoiding the explicit computation of a target-dependent terminal region; see Theorems~\ref{thm:main_stability} and~\ref{col:feasibility_stability_1}. Furthermore, under suitable regularity assumptions on the cost, we introduce a quadratic upper approximation inspired by~\cite{empc_periodicApprox}, enabling a QP formulation that preserves the theoretical guarantees of the general convex formulation while allowing the use of specialized QP solvers.
The proposed approach is validated on a benchmark example and a ball–plate system with eight states, demonstrating its effectiveness and the computational benefit of the QP formulation.

The remainder of the paper is structured as follows. Section~\ref{sec:Sect2} introduces robust optimal periodic tubes and formulates the uncertain periodic operation problem. Section~\ref{sec:Sect3} presents the proposed periodic CCTMPC scheme, establishes its closed-loop properties, and derives its QP specialization. Section~\ref{sec:Sect5} reports numerical results, and Section~\ref{sec:Sect6} concludes the paper with final remarks and future research directions.

\subsubsection*{Notation}
The symbol $\I$ denotes the identity matrix, and \(e_n(i) \in\R^n\) denotes the \(i\)-th canonical basis vector. For vectors $a$ and $b$, we write $(a,b)=[a^{\top}\ b^{\top}]^{\top}$. A block-diagonal matrix with blocks $Q_{1},\ldots,Q_{n}$ is denoted by $\mathrm{blkd}(Q_{1},\ldots,Q_{n})$. The operator $\operatorname{convh}$ denotes the convex hull, and we define the weighted squared norm as $\norm{z}_Q^2 \coloneqq z^{\top}Qz$, where $Q=Q^\top\succ0$. Given a sequence $\mathbf{a}=(a_0,\ldots,a_{T-1})$, we define $\sigma \mathbf a \coloneqq(a_1,\ldots,a_{T-1},a_0)$
as the cyclical shift operator, and apply it componentwise to pairs of sequences, i.e., $\sigma(\mathbf z,\mathbf v)\coloneqq(\sigma \mathbf z,\sigma \mathbf v)$.
 
\section{Problem Formulation} \label{sec:Sect2}
This section formalizes the robust periodic MPC problem for uncertain systems represented as polytopic LDIs~\cite{Boyd1994_LDI_Chapter}. We introduce the system model, define periodic forward invariant tubes, and describe their configuration‑constrained parameterization, leading to the formulation of the optimal periodic tube problem.

\subsection{Uncertain System Model}
We consider uncertain dynamics represented as a polytopic LDI
\begin{equation}\label{eq:system}
    x^+ \in \{ A x + B u + w \;|\; (A,B) \in \Delta,\; w \in \mathcal{W} \},
\end{equation}
where $x \in \R^{n_x}$ and $u \in \R^{n_u}$ denote the state and input, and $w \in \mathcal{W} \subset \R^{n_x}$ is an additive disturbance belonging to a compact convex set. The uncertainty set $\Delta$ is defined as
\[
\Delta \coloneqq \operatorname{convh}\big(\{(A_1,B_1),\ldots,(A_m,B_m)\}\big),
\]
so that $(A,B)$ vary within the convex hull of $m$ vertex pairs. Inclusion~\eqref{eq:system} captures both additive disturbances and parametric uncertainty in a unified framework. Dynamical systems described by LDIs are often used to bound 
Linear Parameter-Varying systems over compact state spaces~\cite{Tth2010}, which is a versatile model class for representing nonlinear systems. Thus, control techniques developed for LDI systems are widely applicable in practice. The system is subject to closed convex constraints $x \in \mathcal{X}$ and $u \in \mathcal{U}$.

To ensure robust operation under all possible realizations of the uncertainty, we adopt a tube‑based MPC approach, in which the controller propagates sets of reachable states (tubes) rather than individual trajectories. The following subsections introduce the concepts required to synthesize such controllers.

\subsection{Periodic Forward Invariant Tubes}
Tube-based MPC relies on the construction of Robust Forward Invariant Tubes (RFITs). To formalize this notion, define
\begin{align}
\label{eq:forward_reach_F}
\mathcal{F}(X):= \left\{ X^+ \subseteq \R^{n_x}
\ \middle| \ \,
\begin{aligned} &\forall x \in X, \ \exists u \in \mathcal{U} : \vspace{1pt} \\
&\forall (A,B)\in\Delta, \ \forall w \in \mathcal{W}, \\
&Ax+Bu+w \in X^+ \vspace{1pt}
\end{aligned} 
\right\}
\end{align}
as the set of all one-step successor sets of a given set $X \subseteq \R^{n_x}$.
\begin{definition}
\label{defn:RFIT}
A sequence of sets $(X_t \subseteq \R^{n_x})_{t \in \mathbb{N}}$ is an RFIT if $X_{t+1} \in \mathcal{F}(X_{t})$ holds for all $t \in \mathbb{N}$.
\end{definition}

Since our control objective involves periodic operation, we focus on the following subclass of RFITs~\cite{LEE20062175,Soyer9304469,SOYER20206905}.
\begin{definition}
\label{defn:PFIT}
An RFIT $(X_t)_{t \in \mathbb{N}}$ is called a 
\emph{Periodic Forward Invariant Tube (PFIT)} 
if there exists a period $T \in \mathbb{N}_{+}$ such that
\begin{align}
\label{eq:setPeriodicityConstr}
X_{t+T} = X_t, \quad \forall t \in \mathbb{N}.
\end{align}
\end{definition}
We remark that a PFIT with period $T=1$ recovers the standard notion for a Robust Control Invariant (RCI) set, since $X_{t+1}=X_t$ implies $X_t \in \mathcal{F}(X_t)$ for all $t$.

\subsection{Configuration-Constrained Polytopic Tubes}
To obtain a tractable representation of PFITs, we adopt the configuration-constrained polyhedral representation introduced in~\cite{CCTMPC_VILLANUEVA2024111543}. A polytope is represented as 
\begin{align}
\label{eq:PolytopicRFIT}
X(y) \coloneqq \{ x \in \mathbb{R}^{n_x} \mid Fx \leq y \},
\end{align}
where $F \in \mathbb{R}^{\fe \times n_x}$ is a fixed matrix of facet normals and $y \in \mathbb{R}^{\fe}$ is a parameter vector. We assume that $Fx \leq 0$ implies $x = 0$ to guarantee that $X(y)$ is bounded (and possibly empty) for any $y$.

To preserve the polytope’s combinatorial structure, $y$ is restricted to the \emph{configuration cone}
\[
\mathcal{E} \coloneqq \{ y \in \mathbb{R}^{\fe} \mid Ey \leq 0 \},
\]
where $E$ encodes linear inequalities ensuring that the adjacency relations between facets remain invariant for all $y \in \mathcal{E}$. Under this condition, there exists a set of matrices $V := \{ V_j \}_{j=1}^{\ve}$ such that
\begin{align}
\label{eq:vertex_config}
X(y) = \mathrm{convh}\big( \{ V_j y \mid j = 1,\ldots,\ve \} \big), \quad \forall y \in \mathcal{E}.
\end{align} 
We refer to $(F,E,V)$ as a \emph{configuration triple}. Procedures for constructing such triples are detailed in~\cite[Sec.~II.A]{houska2024stabilizing}. This structure enables efficient translation between halfspace and vertex-based invariance conditions, which are necessary and sufficient for robust forward invariance~\cite[Prop.~2]{houska2024stabilizing}.
Then, we introduce the set $\mathbb{S} \subseteq \mathbb{R}^{\fe}\times \mathbb{R}^{\ve \cdot n_u}\times\mathbb{R}^{\fe}$ defined as  
\begin{align}
\label{eq:S_definition}
\scalemath{0.95}{
\mathbb{S}\coloneqq
\left\{ (y,u,y^+) \ \middle| \ 
\begin{aligned}
&\forall \,(i,j) \in \{1,\ldots,m\} \times \{1,\ldots,\ve\}, \vspace{1pt}\\
& F(A_i V_j y + B_i u_j) + d \leq y^+, \vspace{1pt}\\
& Ey \leq 0, \ V_j y \in \mathcal{X}, \ u_j \in \mathcal{U} 
\end{aligned} 
\right\},}
\end{align}
where $d_k\coloneqq\max\{F_k w : w\in \mathcal{W}\}$ for all $k \in \{1, \ldots, \mathsf{f} \}$, and $u_j$ denotes the control input associated to vertex $V_jy$. Comparing with the definition of $\mathcal{F}(X)$, observe that the dynamical constraints in the definition of set $\mathbb{S}$ state that $x^+ \in X^+=X(y^+)$ for any $x \in X=X(y)$ and $w \in \mathcal{W}$ using a control input defined by the vertex controls $u_j$. This is formalized in Proposition \ref{prop:RFIT_CC}, which exploits convexity to show that ${\scalemath{0.95}(y,u,y^+) \in \mathbb{S}\Rightarrow X(y^+) \in \mathcal{F}(X(y))}$. Note that we do not enforce any constraints on $y^+$. Rather, the condition only implies that $X(y^+)$ is a reachable set from the constraint-satisfying set $X(y)$.
For simplicity of notation, we denote $u=(u_1,\ldots,u_{\mathsf{v}}) \in \R^{\mathsf{v}n_u}$ as the stacked vertex controls, and introduce $U_j = \I_{n_u} \otimes e_\ve(j)^\top$ such that $u_j = U_j u$.
We work under the following standing assumption on the configuration triple $(F,E,V)$, which ensures that $\mathbb{S} \neq \emptyset$.
\begin{assumption}
      \label{ass:valid_cc}
      There exists some scalar $\gamma \in [0,1)$, and vectors $(y,u,\bar{d})$ with $y>0$, $\bar{d}>d$ satisfying the constraints $Ey \leq 0$, $V_j y \in \mathcal{X}$, $U_j u \in \mathcal{U}$ along with $F(A_i V_j y + B_i U_j u) + \bar{d} \leq \gamma y$  for all $(i,j) \in \{1,\ldots,m\} \times \{1,\ldots,\ve\}$.
\end{assumption}

While there exist several approaches to compute $(F,E,V)$ satisfying Assumption \ref{ass:valid_cc}, see \cite{HOUSKA2025100992}; a detailed description of these approaches is beyond the scope of the current paper. We remark, however, that selecting the \textit{best choice} in terms of computational complexity is an active area of research.

\begin{proposition}
\label{prop:RFIT_CC}
   Let $(y_t)_{t \in \mathbb{N}}$ be a sequence such that for all $t \in \mathbb{N}$ there exists $u_t \in \mathbb{R}^{\ve \cdot n_u}$ with $(y_t,u_t,y_{t+1}) \in \mathbb{S}$. Then, $(X(y_t))_{t \in \mathbb{N}}$ is an RFIT.
\end{proposition}

\textit{Proof:} The result follows from~\cite[Corollary 4]{CCTMPC_VILLANUEVA2024111543}.
\hfill \QED

\subsection{Optimal Configuration-Constrained PFITs} \label{subsec:optPerTraj} 
We consider the problem of computing a PFIT that is 
optimal with respect to a \(T\)-periodic cost. To this end, we denote the tube parameters and the stacked vertex inputs over one period as
\[ \mathbf z\coloneqq(z_0,\ldots,z_{T-1}),\qquad \mathbf v\coloneqq(v_0,\ldots,v_{T-1}) 
\]  
Defining \([l]\coloneqq l\bmod T\) as in~\cite{Zanon2013}, the periodic cost is given by
\begin{align} 
\label{eq:time_varying_econ} 
M_t(\mathbf z,\mathbf v) \coloneqq \sum_{l=0}^{T-1}\ell_{[t+l]}(z_l,v_l), 
\end{align} 
where \(t\in\mathbb N\) denotes global time. Each stage cost may depend on external signals, such as economic weights or reference trajectories. The \emph{optimal PFIT} associated with the periodic cost at time \(t\) is obtained by solving
\begin{align} 
\label{eq:DRTO_problem} 
 (\mathbf z^\circ(t),\mathbf v^\circ(t)) & \in\arg\min_{\mathbf z,\mathbf v}\ M_t(\mathbf z,\mathbf v)\\ 
& \hspace{30pt} \mathrm{s.t.} \begin{cases} (z_l,v_l,z_{l+1})\in\mathbb S, \quad z_T = z_0, \\
\ \forall \ l\in\{0,\ldots,T\!-\!1\}. \end{cases}\nonumber 
\end{align}  
Notice that Assumption \ref{ass:valid_cc} guarantees feasibility of Problem \eqref{eq:DRTO_problem} since the feasible set does not depend on \(t\); only the objective changes through the cyclic shift of the stage costs.
\begin{assumption} 
\label{ass:convex} 
$M_t(\mathbf z,\mathbf v)$ is a strictly convex function.
\end{assumption} 
Assumption \ref{ass:convex} only enforces a strict convexity requirement on the sum of costs $\ell_{[t+l]}$ and not on the individual costs -- They can take any shape as long as Assumption \ref{ass:convex} is verified.
\begin{proposition} 
\label{prop:fixed_cost} 
Suppose Assumption~\ref{ass:convex} holds. Then: \begin{enumerate} 
\item The optimal value of~\eqref{eq:DRTO_problem} is independent of \(t\), i.e.,
\[ \mathcal O^\circ \coloneqq \min_{\mathbf z,\mathbf v}M_t(\mathbf z,\mathbf v), \qquad \forall\,t\in\mathbb N. 
\] 
\item The minimizer in \eqref{eq:DRTO_problem} is unique, and satisfies
\begin{align} 
\label{eq:optimal_PFIT_shift} 
(\mathbf z^\circ(t+1),\mathbf v^\circ(t+1)) = \sigma(\mathbf z^\circ(t),\mathbf v^\circ(t)), \ \ \  \forall\,t\in\mathbb N. \\ \nonumber
\end{align} 
\end{enumerate} 
\end{proposition} 
\begin{proof}
Recalling that $\sigma$ is the cyclical shift operator, for any feasible $(\mathbf{z},\mathbf{v})$ we have $M_{t+1}(\sigma \mathbf{z},\sigma \mathbf{v})=M_t(\mathbf{z},\mathbf{v})$ for all $t \in \mathbb{N}$. Thus, the optimal value is time independent. Uniqueness then follows from Assumption \ref{ass:convex}, and the definition of $\sigma$ implies~\eqref{eq:optimal_PFIT_shift}.
\end{proof} 

By Proposition~\ref{prop:fixed_cost}, the sequence of sets \(\bigl(X(z_0^\circ(t))\bigr)_{t\in\mathbb N}\) coincides with the optimal PFIT. This motivates our control goal.
\paragraph*{Problem Statement} Design an MPC scheme that (i) is recursively feasible and robustly satisfies constraints, and (ii) generates an RFIT \((X(y_t))_{t\in\mathbb N}\) converging to the PFIT optimal for the periodic cost~\eqref{eq:time_varying_econ}. A straightforward approach would be to recompute~\eqref{eq:DRTO_problem} whenever the cost structure changes and then run a tracking MPC for the resulting PFIT. However, such a two-layer design suffers from the drawbacks discussed in the introduction. Instead, following the single-layer MPC philosophy~\cite{LIMON_MPC_for_Track,limon2012model}, the next section proposes a \emph{single convex program} that simultaneously identifies a feasible PFIT and constructs an RFIT to track it, guaranteeing convergence to the optimal PFIT whenever the stage-cost sequence remains unchanged apart from its periodic shift.

\section{CCTMPC for periodic operation} \label{sec:Sect3}
We now present a CCTMPC scheme that addresses our \textit{Problem Statement}. We formulate a parametric convex optimization problem solved online at each time step, in which we compute an RFIT parameterized by vectors $\mathbf{y}=(y_0,\cdots,y_N)$, and a PFIT parameterized by $\mathbf{z}=(z_0,\cdots,z_T)$. The PFIT minimizes the periodic cost in \eqref{eq:time_varying_econ}, while the RFIT simultaneously tracks it. We assume without loss of generality that $T > N$. The optimization problem is
\begin{align}
\min_{\mathbf{y},\mathbf{u},\mathbf{z},\mathbf{v}} \;&\ M_t(\mathbf{z},\mathbf{v}) +\sum_{k=0}^{N-1} \norm{\begin{bmatrix} y_k - z_k \\ u_k-v_k \end{bmatrix}}_Q^2 \!+ \Ze(y_N,\mathbf{z},\mathbf{v}) \notag \\
    & \hspace{-20pt} \text{s.t.} \begin{cases} Fx \leq y_0, \quad  y_N \in \mathbb{T}(\mathbf{z},\mathbf{v}), \quad z_T = z_0,\\
   (y_k, u_k, y_{k+1}) \in \mathbb{S},\  \forall \, k \in \{0,\ldots,N\!-\!1\},\\
    (z_l, v_l, z_{l+1}) \in \mathbb{S},\  \forall \, l \in \{0,\ldots,T\!-\!1\}. \end{cases} \label{eq:periodic_MPC}
\end{align}   
which is parametric in current state $x=x_t$ and periodic cost $M_t$. Note that the state and input constraints are captured by $\mathbb{S}$, and the variables  $\mathbf u=(u_0,\ldots,u_{N-1})$ and $\mathbf v=(v_0,\ldots,v_{T-1})$ denote the vertex control inputs.
The PFIT characterized by $\mathbf{z}$ is \textit{artificial}, in the sense that it acts as a vehicle to link the RFIT in which the real system operates to the optimal PFIT, which is the goal of the operation. The physical state is linked to the RFIT by \(Fx\leq y_0\). Furthermore, \(\Ze\) denotes the terminal cost and \(\mathbb T\) the associated terminal region, both defined relative to the artificial PFIT. Their construction is detailed next.

\subsection{Design of the terminal ingredients} 
We define the terminal cost $\Ze$ and the induced terminal set $\mathbb T$ as a finite‑horizon approximation of the infinite‑horizon control law. To this end, we introduce the cost‑to‑travel function~\cite{VILLANUEVA2020109179}
\begin{equation}
\label{eq:cost_to_go}
\begin{aligned}
    \Ve(y,y^+,z,v):= \min_u \,&\ \lVert(y-z,u-v)\rVert_Q^2 \\
    \text{s.t.} \;&\ (y,u,y^+) \in \mathbb{S},
\end{aligned}
\end{equation}
with $\Ve(y,y^+,z,v)=\infty$ if~\eqref{eq:cost_to_go} is infeasible. Note that $(y,y^+,z,v)$ are fixed in \eqref{eq:cost_to_go}, with the objective minimizing the control effort to transition between the sets.
Then, Problem~\eqref{eq:periodic_MPC} is equivalent to
\begin{align}
\min_{\mathbf{y},\mathbf{z},\mathbf{v}}\, &\ M_t(\mathbf{z},\mathbf{v}) + \!\sum_{k=0}^{N-1} \Ve(y_k,y_{k+1},z_k,v_k) + \Ze(y_N,\mathbf{z},\mathbf{v}) \notag\\
&\hspace{-15pt} \text{s.t.} \begin{cases} Fx \leq y_0, \quad y_N \in \mathbb{T}(\mathbf{z},\mathbf{v}), \quad z_T = z_0, 
\\
 (z_l, v_l, z_{l+1}) \in \mathbb{S},\  \forall \, l \in \{0,\ldots,T\!-\!1\}. \end{cases}
\label{eq:periodic_MPC_2}
\end{align}
Note that the RFIT 
constraints $(y_k,u_k,y_{k+1}) \in \mathbb{S}$ are embedded in $\Ve$.
Then, we define the terminal cost as
\begin{align} 
        \Ze(y,\mathbf{z},\mathbf{v}):=&\min_u \  \lVert(y-z_N,u-v_N)\rVert_P^2  \label{eq:terminal_cost}\\
       &\hspace{-50pt} \text{s.t.}\begin{cases} F(A_iV_j(y-z_N)+B_iU_j(u-v_N)) \leq \gamma(y-z_N),  \\
        V_j(z_{N+l}+\gamma^l (y-z_N)) \in \mathcal{X},  \\ 
          U_j(v_{N+l}+\gamma^l(u-v_N)) \in \mathcal{U},  \\
        E(z_{N+l}+\gamma^l(y-z_N)) \leq 0,  \\
      \ \forall \,(i,j,l) \in \{1,\ldots,m\} \times \{1,\ldots,\ve\} \times \{0,\ldots,T\!-\!1\},   \end{cases} \nonumber
\end{align}
for some $\gamma \in [0,1)$. Intuitively, the constraints of \eqref{eq:terminal_cost} enforce a contractive condition \textit{about the PFIT}, and are a generalization of the contractive condition towards an RCI set \cite{CCTMPC_Tracking_10546969}.
The terminal region is implicitly defined as \( \mathbb T(\mathbf z,\mathbf v) \coloneqq \{y\mid\Ze(y,\mathbf z,\mathbf v)<\infty\}\). We next show that \(\Ze\) satisfies a decrease condition on \(\mathbb T(\mathbf z,\mathbf v)\).
\begin{theorem} 
\label{thm:main_stability} 
Suppose Assumption~\ref{ass:convex} holds, \((\mathbf z,\mathbf v)\) satisfies 
\begin{align} 
\label{eq:z_constraints}
(z_l,v_l,z_{l+1})\in\mathbb S, \qquad \forall\,l\in\{0,\ldots,T-1\},
\end{align} 
and matrices $Q$ and $P$ along with scalar $\gamma \in [0,1)$ satisfy
\begin{align} 
\label{eq:gamma_condition} 
Q \succ 0, && P \succ 0, && \gamma^2P+Q\preceq P, 
\end{align} 
Then, \(\Ze(y,\mathbf z,\mathbf v)\) is positive definite about \(y=z_N\), and for every \(\bar y\in\mathbb T(\mathbf z,\mathbf v)\), there exists \(\bar y^+\in\mathbb T(\sigma\mathbf z,\sigma\mathbf v)\) such that 
\begin{align} 
\label{eq:Lyapunov_decease} 
\Ze(\bar y^+,\sigma\mathbf z,\sigma\mathbf v) + \Ve(\bar y,\bar y^+,z_N,v_N) \leq \Ze(\bar y,\mathbf z,\mathbf v).  \\ \nonumber
\end{align} 
\end{theorem}
\begin{proof} 
Positive definiteness of $\Ze$ about $y=z_N$ follows from convexity of cost and linearity of constraints. 
Let $\bar{u}$ optimize Problem~\eqref{eq:terminal_cost} at $\bar y\in\mathbb T(\mathbf z,\mathbf v)$, such that 
\begin{align} 
\label{eq:Z_current} \Ze(\bar y,\mathbf z,\mathbf v) = \norm{(\bar y-z_N,\bar u-v_N)}_P^2. 
\end{align} 
By~\eqref{eq:z_constraints}, $F(A_iV_jz_N+B_iU_jv_N)+d \leq z_{N+1}$
holds for every $(i,j)$.
Adding the first constraint in~\eqref{eq:terminal_cost} and defining $\bar y^+ \coloneqq z_{N+1}+\gamma(\bar y-z_N)$
shows that \(\bar u\) is feasible for $\Ve(\bar y,\bar y^+,z_N,v_N)$. Hence,
\begin{align} 
\label{eq:V_current} 
\Ve(\bar y,\bar y^+,z_N,v_N) \leq \norm{(\bar y-z_N,\bar u-v_N)}_Q^2. 
\end{align}
We next show that $\bar y^+\in\mathbb T(\sigma\mathbf z,\sigma\mathbf v)$. By definition of operator $\sigma$, $\sigma(\mathbf{z},\mathbf{v})_N = (z_{N+1},v_{N+1})$ and $\sigma(\mathbf{z},\mathbf{v})_{N+l} = (z_{N+l+1},v_{N+l+1})$.
Consider the candidate $\bar u^+ \coloneqq v_{N+1}+\gamma(\bar u-v_N)$.
Then, 
\begin{align*}
    & F( A_iV_j(\bar y^+-(\sigma\mathbf z)_N) + B_iU_j(\bar u^+-(\sigma\mathbf v)_N))  \\
    &\hspace{40pt} = \gamma F( A_iV_j(\bar y-z_N) + B_iU_j(\bar u-v_N) )  \\
    & \hspace{60pt}  \leq \gamma^2(\bar y-z_N) = \gamma ( \bar y^+-(\sigma\mathbf z)_N ),
\end{align*}  
follows for all $(i,j)$ since $(\bar{y},\bar{u})$ is the optimizer of $\Ze(\bar{y},\mathbf{z},\mathbf{v})$. 
    For the state constraints, $\gamma^l ( \bar y^+-(\sigma\mathbf z)_N)=\gamma^{l+1}(\bar y-z_N)$ implies that
    \begin{align*} V_j( (\sigma\mathbf z)_{N+l} + \gamma^l ( \bar y^+-(\sigma\mathbf z)_N)) \in\mathcal X, \ l \in \{0,\ldots,T-2\}
    \end{align*}
    by feasibility at $y=\bar{y}$. Observing that \((\sigma\mathbf z)_{N+T-1}=z_N\), we have
     \begin{align*} 
    V_j( (\sigma\mathbf z)_{N+T-1} &+ \gamma^{T-1} ( \bar y^+-(\sigma\mathbf z)_N ) ) \\ 
    &\hspace{-45pt} = V_j( z_N+\gamma^T(\bar y-z_N)) = \gamma^TV_j\bar y + (1-\gamma^T)V_jz_N \in\mathcal X,
    \end{align*} 
    by convexity. 
    The input and configuration constraints follow from the same argument, using \(\bar u^+-(\sigma\mathbf v)_N=\gamma(\bar u-v_N)\), and convexity of \(\mathcal U\) and \(\mathcal E\). Therefore, \(\bar u^+\) is feasible for \(\Ze(\bar y^+,\sigma\mathbf z,\sigma\mathbf v)\), and 
    \begin{align} 
    \label{eq:Z_next} 
    \Ze(\bar y^+,\sigma\mathbf z,\sigma\mathbf v) \leq \gamma^2 \norm{(\bar y-z_N,\bar u-v_N)}_P^2. 
    \end{align} 
    Combining~\eqref{eq:Z_current}, \eqref{eq:V_current}, \eqref{eq:Z_next}, and~\eqref{eq:gamma_condition} proves \eqref{eq:Lyapunov_decease}. 
    \end{proof}

In summary, Problem~\eqref{eq:periodic_MPC} can be implemented as  
\begin{equation}
\begin{aligned}
\min_{\mathbf{y},\mathbf{u},\mathbf{z},\mathbf{v}} &\
M_t(\mathbf{z},\mathbf{v}) +\!\!\sum_{k=0}^{N-1} \norm{\begin{bmatrix} y_k - z_k \\ u_k-v_k \end{bmatrix}}_Q^2 \!\!+  
\norm{\begin{bmatrix} y_N - z_N \\ u_N-v_N \end{bmatrix}}_P^2  \label{eq:periodic_MPC_implementation_1} \\
& \hspace{-35pt} \text{s.t.} \begin{cases} 
Fx \leq y_0, \quad z_T = z_0,
\\
 (y_k, u_k, y_{k+1}) \in \mathbb{S},\quad  \forall \, k \in \{0,\ldots,N\!-\!1\},\\
(z_l, v_l, z_{l+1}) \in \mathbb{S},\quad  \forall \, l \in \{0,\ldots,T\!-\!1\}, \\
\hdashline
F(A_iV_j(y_N-z_N)+B_iU_j(u_N-v_N)) \leq \gamma(y_N-z_N),  \\
V_j(z_{N+{\bar{l}}}+\gamma^{\bar{l}} (y_N-z_N)) \in \mathcal{X},  \\ 
  U_j(v_{N+{\bar{l}}}+\gamma^{\bar{l}}(u_N-v_N)) \in \mathcal{U},  \\
E(z_{N+{\bar{l}}}+\gamma^{\bar{l}}(y_N-z_N)) \leq 0,  \\
\ \forall \,(i,j,\bar{l}) \in \{1,\ldots,m\} \times \{1,\ldots,\ve\} \times \{0,\ldots,T\!-\!1\}, 
\end{cases}
\end{aligned}
\end{equation}
where the first three constraints denote the usual RFIT and PFIT constraints, and the remaining are terminal constraints from \eqref{eq:terminal_cost}.

\begin{remark}
\label{remark:interpretation}
Sufficient conditions for the constraints of Problem \eqref{eq:terminal_cost} can be derived in a manner similar to \cite{houska2024stabilizing} as
\begin{subequations}
    \label{eq:terminal_sufficient}
\begin{align}
    &V_j z_l \in \epsilon \mathcal{X}, \ U_j v_l \in \epsilon \mathcal{U}, \ E(y-z_N) \leq 0,\\
    &F(A_iV_j(y-z_N)+B_iU_j(u-v_N)) \leq \gamma (y-z_N),     \label{eq:terminal_sufficient:1}\\
    & V_j(y-z_N) \in (1-\epsilon) \mathcal{X}, \ U_j(u-v_N) \in (1-\epsilon) \mathcal{U}, \\
    & \forall \,(i,j,l) \in \{1,\ldots,m\} \times \{1,\ldots,\ve\} \times \{0,\ldots,T\!-\!1\}
\end{align}
\end{subequations}
with some large $\epsilon \in (0,1)$, such that $X(y-z_N)$ is a $\gamma$-contractive configuration constrained polytope within the tightened constraints defined using $(1-\epsilon)$. The satisfaction of \eqref{eq:Lyapunov_decease} when the constraints of Problem \eqref{eq:terminal_cost} are replaced by \eqref{eq:terminal_sufficient} can be verified using similar steps as in the proof of Theorem \ref{thm:main_stability}, and are not reported here for brevity. Essentially, \eqref{eq:terminal_sufficient} dictates that $y$ belongs in a neighborhood of $z_N$ from which $y-z_N$ persistently decays, ensuring constraint satisfaction along the whole PFIT trajectory. While viable, however, \eqref{eq:terminal_sufficient} might render a very small region of attraction. Hence, we proceed with the constraints in \eqref{eq:terminal_cost} for numerical illustrations. 
\end{remark}

\subsection{Closed-loop control law}
\label{sec:closed_loop_control_law}
We now define a closed-loop control law that links the parameter space of the RFIT to the statespace of \eqref{eq:system} through a control input. Denote the optimizer of Problem~\eqref{eq:periodic_MPC_implementation_1} at state $x=x_t$ by $(\mathbf{y}^{\star}(x_t),\mathbf{u}^{\star}(x_t),\mathbf{z}^{\star}(x_t),\mathbf{v}^{\star}(x_t))$, and the
optimal stage cost by $\mathrm{V}_k^{\star}(x_t)=\mathrm{V}(y_k^{\star}(x_t), y_{k+1}^{\star}(x_t), z_k^{\star}(x_t), v_k^{\star}(x_t))$.
The control law is defined using the vertex inputs of the set $X(y_0^{\star}(x_t))$ as
\begin{align}
\label{eq:MPC_control_law}
\mu(x_t) \coloneqq \sum_{j=1}^{\ve} \lambda_j(x_t)\,u^{\star}_{0,j}(x_t),
\end{align}
where $\lambda(x_t) \in \R^{\ve}$ solves the QP
\begin{align*}
\min_{\lambda \ge 0} \ \norm{\lambda}_2^2 \quad
\text{s.t.} \quad
x_t = \sum_{j=1}^{\ve} \lambda_j\, V_j y^{\star}_0(x_t), \qquad
\sum_{j=1}^{\ve} \lambda_j = 1.
\end{align*}
Under~\eqref{eq:MPC_control_law}, the closed-loop system evolves according to
\begin{align}
\label{eq:CL_system}
 x_{t+1} \in \{ A x_t + B\mu(x_t) + w \mid (A,B)\in\Delta, w\in\mathcal{W}\}.
\end{align}
Let $\mathcal{O}(x_t)$ denote the optimal value of Problem~\eqref{eq:periodic_MPC_implementation_1} at time $t$.
\begin{assumption}
    \label{ass:interiority}
    The optimal PFIT parameters belong strictly in the interior of the constraints, i.e., there exists some $\epsilon \in (0,1)$ such that $V_j z_{l}^{\star}(x) \in \epsilon \mathcal{X}$, $U_j v_{l}^{\star}(x) \in \epsilon \mathcal{U}$ and $E z_{l}^{\star}(x) \leq -(1-\epsilon) \mathbf{1}$ for all $j\in\{1,\ldots,\ve\}$, $l\in\{0,\ldots,T-1\}$, and all $x$ feasible for Problem \eqref{eq:periodic_MPC_implementation_1}, with the same conditions satisfied by $(\mathbf{z}^{\circ},\mathbf{v}^{\circ})$.
\end{assumption}

Assumption~\ref{ass:interiority} is a sufficient condition commonly found in the literature, e.g.,~\cite{LIMON_Periodic_7172461}. Boundary optimal PFITs can also be accommodated when their active constraints admit the feasible directions required by the local optimality argument in the theorem below. We retain the stated condition to avoid active-set-dependent assumptions.

\begin{theorem} \label{col:feasibility_stability_1}
    Suppose Assumptions~\ref{ass:convex} and \ref{ass:interiority} hold, and~\eqref{eq:gamma_condition} is satisfied. Then, for every trajectory of~\eqref{eq:CL_system}: 
\begin{enumerate}
    \item If~\eqref{eq:periodic_MPC_implementation_1} is feasible at \(t=0\), it remains feasible for all \(t\in\mathbb{N}\).
    \item If $M_{t+T}   (\mathbf{z},\mathbf{v})=M_t(\mathbf{z},\mathbf{v})$ for all $t \in \mathbb{N}$, then
    \begin{align}
    \label{eq:Lyapunov_decrease}
       \mathcal{O}(x_{t+1}) - \mathcal{O}(x_t) \leq - \mathrm{V}^{\star}_0(x_t)
    \end{align}
    for any $x_{t+1}$ satisfying \eqref{eq:CL_system}, and for all $k \in \{0,\ldots,N\}$ we have
    $\scalemath{0.98}{\lim_{t \to \infty} (y^{\star}_k(x_t),u^{\star}_k(x_t)) = (z^{\star}_k(x_t),v^{\star}_k(x_t))}$.
     \item The RFIT converges to the optimal PFIT parameters as $\scalemath{0.98}{\lim_{t \to \infty} (y^{\star}_k(x_t),u^{\star}_k(x_t))=(z^{\circ}_k(t),v^{\circ}_k(t))}$, such that
     \begin{align}
\label{eq:Lyapunov_function}
\mathcal{L}(x_t) \coloneqq \mathcal{O}(x_t) - \mathcal{O}^\circ
\end{align}
serves as a Lyapunov function.
\end{enumerate}
\end{theorem}
\begin{proof}
 Statements $(1)$ and $(2)$ follow from Theorem \ref{thm:main_stability}, which constructs a feasible point in the terminal set $\mathbb{T}(\mathbf{z}^{\star}(x_t),\mathbf{v}^{\star}(x_t))$ to demonstrate that the terminal cost serves as a valid Lyapunov function. For further details, we refer to~\cite[Cor.~1]{CCTMPC_Tracking_10546969}.
We divide the proof of Statement $(3)$ into three parts.
\\
\textit{\textbf{Part 1:}} Suppose the optimizers of \eqref{eq:periodic_MPC_implementation_1} converge as in Statement $(2)$ to the zero-decrease set. By contradiction, we show that this set is singleton, i.e., it only consists the optimal PFIT. Since we analyze this optimizer at a fixed time step, we drop time dependence. \\
Suppose that the PFIT satisfies
\begin{align}
    \label{eq:false_assumption}
    (\mathbf{z}^{\star},\mathbf{v}^{\star}) \neq (\mathbf{z}^{\circ},\mathbf{v}^{\circ}). 
\end{align}
Lets define $ (\tilde{\mathbf{z}},\tilde{\mathbf{v}}) = (\mathbf{z}^{\star},\mathbf{v}^{\star})+\lambda((\mathbf{z}^{\circ},\mathbf{v}^{\circ})-(\mathbf{z}^{\star},\mathbf{v}^{\star}))$ for some $\lambda \in [0,1]$, which satisfies the convex PFIT constraints. Denoting the perturbation $\mathbf{a}=M(\mathbf{z}^{\star},\mathbf{v}^{\star})-\mathcal{O}^{\circ}$, Assumption \ref{ass:convex} implies that
\begin{align}
    \label{eq:cost_ineq}
    M(\tilde{\mathbf{z}},\tilde{\mathbf{v}}) \leq M(\mathbf{z}^{\star},\mathbf{v}^{\star})-\lambda\mathbf{a}, && \forall \lambda \in [0,1]
\end{align}
by convexity, with $\mathbf{a}>0$ under \eqref{eq:false_assumption} by uniqueness of the optimizer. \\
\textit{\textbf{Part 2:}} Since the optimizers converge, we know that
\begin{align}
    \label{eq:MPC_cost_opt}
    \mathcal{O} = M(\mathbf{z}^{\star},\mathbf{v}^{\star}).
\end{align}
Now, we construct candidate RFIT parameters $(\tilde{\mathbf{y}},\tilde{\mathbf{u}})$ corresponding to PFIT parameters $(\tilde{\mathbf{z}},\tilde{\mathbf{v}})$ to derive an upper bound to $\mathcal{O}$. We define the candidate RFIT with $\tilde{y}_k=y^{\star}_k=z^{\star}_k$ and $\tilde{u}_k=u^{\star}_k=v^{\star}_k$ for $k \in \{0,\ldots,N-1\}$. Then, for some vector $\bar{s}>0$, suppose that
\begin{align}
    \label{eq:terminal_RCI}
    \exists (\bar{y},\bar{u})  :  E \bar{y} \leq 0, \bar{y} \geq 0,\  F(A_iV_j \bar{y}+B_i U_j \bar{u}) \leq \gamma \bar{y}-\bar{s}.
\end{align}
Using $\bar{y}$, we define $\tilde{y}_N=z^{\star}_N+\lambda \bar{y}$ for some $\lambda \in [0,1]$. Under Assumption \ref{ass:valid_cc}, \eqref{eq:terminal_RCI} is validated for any $\bar{s} \geq 0$ by scaling the vectors. Clearly, the candidate RFIT $(\tilde{\mathbf{y}},\tilde{\mathbf{u}})$ satisfies the initial and path constraints of \eqref{eq:periodic_MPC_implementation_1}, with $(\tilde{y}_{N-1},\tilde{u}_{N-1},\tilde{y}_N) \in \mathbb{S}$ being satisfied since $\tilde{y}_N \geq y^{\star}_N$ for all $\lambda \in [0,1]$. Now, we show that $\tilde{y}_N \in \mathbb{T}(\tilde{\mathbf{z}},\tilde{\mathbf{v}})$. \\
Denoting $e^z_l=z^{\star}_l-z^{\circ}_l$ and $e^v_l = v^{\star}_l-v^{\circ}_l$, the terminal constraint holds according to \eqref{eq:terminal_cost} if there exists some $u \in \R^{\ve n_u}$ such that
\begin{subequations}
\label{eq:terminal_to_hold_2}
\begin{align}
   &  F(\lambda A_iV_j(\bar{y}+e^z_N)+B_iU_j(u-\tilde{v}_N)) \leq \gamma \lambda (\bar{y}+e^z_N), \label{eq:terminal_to_hold_2:1}\\
        & V_jz^{\star}_{[N+l]}+\lambda(-V_j e^z_{[N+l]} + \gamma^l V_j(\bar{y}+e^z_N)) \in \mathcal{X},   \label{eq:terminal_to_hold_2:2} \\ 
        &  U_jv^{\star}_{[N+l]}-\lambda U_je^v_{[N+l]}+\gamma^lU_j(u-\tilde{v}_N) \in \mathcal{U},   \label{eq:terminal_to_hold_2:3} \\
        &   Ez^{\star}_{[N+l]}+\lambda (-Ee^z_{[N+l]}+\gamma^l E(\bar{y}+e^z_N)) \leq 0  \label{eq:terminal_to_hold_2:4}
\end{align}
\end{subequations} 
for all appropriate indices $(i,j,l)$.
Lets parameterize $u = \tilde{v}_N+\lambda \rho$, such that $u-\tilde{v}_N=\lambda \rho$. Substituting in \eqref{eq:terminal_to_hold_2:3}, observe that under Assumption \ref{ass:interiority}, there exists some $\bar{\lambda} \in (0,1]$ such that \eqref{eq:terminal_to_hold_2:2}-\eqref{eq:terminal_to_hold_2:4} are satisfied for all $\lambda \in [0,\bar{\lambda}]$ since the optimal PFIT belongs in the interior of the convex constraints. Then, \eqref{eq:terminal_to_hold_2:1} is the only constraint remaining to validate, which we write for $\lambda \in (0,\bar{\lambda}]$ as
\begin{align}
    \label{eq:terminal_to_hold_3}
    F(A_i V_j \bar{y} + B_i U_j \rho) \leq \gamma(\bar{y}+e^z_N)- F A_i V_j e^z_N.
\end{align}
Note that \eqref{eq:terminal_to_hold_2:1} holds trivially if $\lambda=0$. Comparing \eqref{eq:terminal_RCI} and \eqref{eq:terminal_to_hold_3}, it follows that $\rho=\bar{u}$ satisfies \eqref{eq:terminal_to_hold_3} if
\begin{align}
    \label{eq:s_condition}
    \bar{s} \geq FA_iV_j e^z_N-\gamma e^z_N, \forall (i,j) \in \{1,\ldots,m\} \times \{1,\ldots,\ve\}.
\end{align}
Since $(\bar{y},\bar{u})$ satisfying \eqref{eq:terminal_RCI} exist for any $\bar{s}$ satisfying \eqref{eq:s_condition} under Assumption \ref{ass:valid_cc} by appropriate scaling since the state and input constraints are absent, it follows that $\tilde{y}_N=z^{\star}_N+\lambda \bar{y}$ is feasible, with corresponding input vector $\tilde{u}_N = \tilde{v}_N+\lambda \bar{u}$ for all $\lambda \in [0,\bar{\lambda}]$. \\
\textit{\textbf{Part 3:}} Using the feasible candidate $(\tilde{\mathbf{y}},\tilde{\mathbf{u}},\tilde{\mathbf{z}},\tilde{\mathbf{v}})$, we define
\begin{align}
    \label{eq:feasible_tracking}
    \mathbf{b}=\sum_{k=0}^{N-1} \norm{(e^z_k,e^v_k)}_Q^2 + \norm{(\bar{y}+e^z_N,\bar{u})}_P^2,
\end{align}
such that the cost along the feasible trajectory is $M(\tilde{\mathbf{z}},\tilde{\mathbf{v}})+\lambda^2 \mathbf{b}$ for all $\lambda \in [0,\bar{\lambda}]$. Thus, by optimality of MPC, it follows that
\begin{align}
    \label{eq:optimality_bnd}
    \mathcal{O} \leq M(\tilde{\mathbf{z}},\tilde{\mathbf{v}})+\lambda^2 \mathbf{b}, && \forall \ \lambda \in [0,\bar{\lambda}].
\end{align}
Then, from \eqref{eq:cost_ineq} and \eqref{eq:MPC_cost_opt}, it follows that
\begin{align}
    \label{eq:contradiction_bnd}
    \mathcal{O} \leq \mathcal{O}-\lambda \mathbf{a}+\lambda^2 \mathbf{b}, && \forall \ \lambda \in [0,\bar{\lambda}].
\end{align}
Under \eqref{eq:false_assumption}, we know that $\mathbf{b}>0$. Then, \eqref{eq:contradiction_bnd} is a contradiction, since $-\lambda \mathbf{a}+\lambda^2 \mathbf{b}<0$ whenever $\lambda \in (0,\mathbf{a}/\mathbf{b})$. Hence, we cannot have $\mathbf{a}>0$, which under Assumption \ref{ass:convex} implies that \eqref{eq:false_assumption} cannot occur. Thus, the PFIT optimizers $(\mathbf{z}^{\star},\mathbf{v}^{\star})=(\mathbf{z}^{\circ},\mathbf{v}^{\circ})$, which along with Statement $(2)$ concludes the proof of the first part of Statement $(3)$. Then, the observation $\mathcal{O}>\mathcal{O}^{\circ}$ whenever the optimizers have not converged and $\mathcal{O}=\mathcal{O}^{\circ}$ otherwise, along with \eqref{eq:Lyapunov_decrease} conclude the proof that $\mathcal{L}$ is a Lyapunov function for \eqref{eq:CL_system}.
\end{proof}

\subsection{QP approximation of the periodic cost} 
\label{subsec:QP_approximation} 
When \(M_t\) is nonquadratic, Problem~\eqref{eq:periodic_MPC_implementation_1} is convex but not, in general, a QP. Following~\cite{empc_periodicApprox}, we replace \(M_t\) with a quadratic upper bound model while retaining the same constraints.
\begin{assumption} 
\label{ass:quadratic_majorizer} 
There exists some scalar $\rho>0$ such that, for every \(t\in\mathbb N\) and all feasible \((\mathbf z,\mathbf v)\) and \((\hat{\mathbf z},\hat{\mathbf v})\) of Problem~\eqref{eq:DRTO_problem},
\begin{multline}
\scalemath{0.97}{M_t(\mathbf z,\mathbf v) \leq M_t(\hat{\mathbf z},\hat{\mathbf v}) + \nabla M_t(\hat{\mathbf z},\hat{\mathbf v})^\top  \begin{bmatrix}\Delta \mathbf{z} \\ \Delta \mathbf{v}\end{bmatrix}+ \frac{\rho}{2} \norm{ \begin{bmatrix} \Delta \mathbf{z} \\ \Delta\mathbf{v}\end{bmatrix}}_{2}^{2}}
\label{eq:quadratic_upper_bound} 
\end{multline} 
where $(\Delta \mathbf{z},\Delta \mathbf{v}) = (\mathbf{z}-\hat{\mathbf{z}},\mathbf{v}-\hat{\mathbf{v}})$.
\end{assumption}

Assumption \ref{ass:quadratic_majorizer} holds if the Hessian of $M_t$ is bounded over its domain. Denoting $\hat{\mathbf{q}} = (\hat{\mathbf{z}},\hat{\mathbf{v}})$, we label  the right hand side of~\eqref{eq:quadratic_upper_bound} by $\hat M_t( \mathbf z,\mathbf v; \hat{\mathbf q})$ such that $M_t(\mathbf z,\mathbf v)\leq \hat M_t(\mathbf z,\mathbf v;\hat{\mathbf q})$. The MPC controller is then formulated based on the QP
\begin{align}
\label{eq:periodic_MPC_implementation_2}
\min_{\mathbf y,\mathbf u,\mathbf z,\mathbf v}\ \ & \hat M_t( \mathbf z,\mathbf v; \hat{\mathbf q}) + \sum_{k=0}^{N-1} \norm{ \begin{bmatrix} y_k-z_k\\ u_k-v_k \end{bmatrix} }_Q^2\!\! + \norm{ \begin{bmatrix} y_N-z_N\\ u_N-v_N \end{bmatrix} }_P^2 \notag\\
\mathrm{s.t.}\quad & \text{constraints from~\eqref{eq:periodic_MPC_implementation_1}} 
\end{align} 
by replacing the objective of \eqref{eq:periodic_MPC_implementation_1} with its quadratic approximation. We denote the parametric optimizer of \eqref{eq:periodic_MPC_implementation_2} as $(\mathbf y^\bullet(x,\hat{\mathbf{q}}),\mathbf u^\bullet(x,\hat{\mathbf{q}}), \mathbf z^\bullet(x,\hat{\mathbf{q}}),\mathbf v^\bullet(x,\hat{\mathbf{q}}))$, and optimal value as $\hat{\mathcal{O}}(x,\hat{\mathbf{q}})$. The closed-loop control law is derived in the same way as in Section \ref{sec:closed_loop_control_law}, with component $\lambda$ computed based on $y_0^{\bullet}(x_t,\hat{\mathbf{q}}_t)$. The linearization point is updated using the PFIT trajectory as
\begin{align} 
\label{eq:linearization_update} 
\hat{\mathbf q}_{t+1} = \sigma(\mathbf z^\bullet(x_t,\hat{\mathbf{q}}_t),\mathbf v^\bullet(x_t,\hat{\mathbf{q}}_t)), 
\end{align} 
and is initialized with any feasible PFIT.

\begin{corollary} 
\label{cor:QP_properties} 
Suppose Assumptions~\ref{ass:convex}, \ref{ass:interiority}, and~\ref{ass:quadratic_majorizer} hold, and~\eqref{eq:gamma_condition} is satisfied. Then:
\begin{enumerate} 
\item If~\eqref{eq:periodic_MPC_implementation_2} is feasible at \(t=0\), it remains feasible for all \(t\in\mathbb N\). 
\item If $M_{t+T}(\mathbf{z},\mathbf{v})=M_t(\mathbf{z},\mathbf{v})$ for all $t \in \mathbb{N}$, then 
\begin{align} 
\label{eq:QP_value_decrease_2} 
&\hat{\mathcal O}(x_{t+1},\hat{\mathbf{q}}_{t+1})-\hat{\mathcal O}(x_{t},\hat{\mathbf{q}}_{t}) \\
& \hspace{20pt} \leq -\Ve( y_{0}^\bullet(x_{t},\hat{\mathbf{q}}_{t}), y_{1}^\bullet(x_{t},\hat{\mathbf{q}}_{t}), z_{0}^\bullet(x_{t},\hat{\mathbf{q}}_{t}), v_{0}^\bullet(x_{t},\hat{\mathbf{q}}_{t}) ), \nonumber
\end{align} 
with $\scalemath{0.97}{\lim_{t \to \infty} (y^{\bullet}_k(x_t,\hat{\mathbf{q}}_t),u^{\bullet}_k(x_t,\hat{\mathbf{q}}_t)) = (z^{\bullet}_k(x_t,\hat{\mathbf{q}}_t),v^{\bullet}_k(x_t,\hat{\mathbf{q}}_t))}$.
\item The RFIT converges to the optimal PFIT parameters as $\scalemath{0.98}{\lim_{t \to \infty} (y^{\bullet}_k(x_t,\hat{\mathbf{q}}_t),u^{\bullet}_k(x_t,\hat{\mathbf{q}}_t))=(z^{\circ}_k(t),v^{\circ}_k(t))}$, such that 
\end{enumerate} 
\begin{align} 
\hat{\mathcal L}_t \coloneqq \hat{\mathcal O}(x_t,\hat{\mathbf{q}}_t)-\mathcal O^\circ \label{eq:approxLyapunovFunction} 
\end{align} 
is a Lyapunov function for the QP-based formulation.
\end{corollary} 
\begin{proof} 
Problems~\eqref{eq:periodic_MPC_implementation_1} and \eqref{eq:periodic_MPC_implementation_2} have the same feasible set. Hence, recursive feasibility follows from the shifted candidate used in the proof of Theorem~\ref{col:feasibility_stability_1}. The PFIT component of this candidate at time $t+1$ is $\sigma(\mathbf{z}^{\bullet}(x_t,\hat{\mathbf{q}}_t),\mathbf{v}^{\bullet}(x_t,\hat{\mathbf{q}}_t))=\hat{\mathbf{q}}_{t+1}$. Similar to \cite{empc_periodicApprox}, we consider the auxiliary linearization point $\tilde{\mathbf{q}}_{t+1}=\sigma \hat{\mathbf{q}}_t$. Denoting the dependence of the optimizers on $(x_t,\hat{\mathbf{q}}_t)$ by subscript $t$, from \eqref{eq:linearization_update} and Assumption \ref{ass:quadratic_majorizer}, it follows that
\begin{align*}
    \hat{M}_{t+1}(\sigma\mathbf{z}^{\bullet}_t,\sigma\mathbf{v}^{\bullet}_t;\hat{\mathbf{q}}_{t+1}) &= M_{t+1}(\sigma\mathbf{z}^{\bullet}_t,\sigma\mathbf{v}^{\bullet}_t) \\
    & \hspace{30pt} \leq \hat{M}_{t+1}(\sigma\mathbf{z}^{\bullet}_t,\sigma\mathbf{v}^{\bullet}_t;\tilde{\mathbf{q}}_{t+1}). \nonumber
\end{align*}
From the definitions of $\hat{M}_t$ in \eqref{eq:quadratic_upper_bound} and $\tilde{\mathbf{q}}_{t+1}$, the approximated PFIT cost satisfies $\hat{M}_{t+1}(\sigma \mathbf{z}^{\bullet}_t,\sigma \mathbf{v}^{\bullet}_t;\tilde{\mathbf{q}}_{t+1}) = \hat{M}_t(\mathbf{z}^{\bullet}_t,\mathbf{v}^{\bullet}_t;\hat{\mathbf{q}}_t)$ resulting in
$\hat{M}_{t+1}(\sigma\mathbf{z}^{\bullet}_t,\sigma\mathbf{v}^{\bullet}_t;\hat{\mathbf{q}}_{t+1}) \leq \hat{M}_t(\mathbf{z}^{\bullet}_t,\mathbf{v}^{\bullet}_t;\hat{\mathbf{q}}_t)$. This inequality along with Theorem \ref{thm:main_stability} implies \eqref{eq:QP_value_decrease_2}. The convergence of optimizers follows from positive-definiteness of the tracking cost.

Regarding Statement 3, we note from \eqref{eq:linearization_update} that the convergence of optimizers in Statement 2 implies $\lim_{t \to \infty} \hat{\mathbf{q}}_t=(\mathbf{z}^{\bullet}_t,\mathbf{v}^{\bullet}_t)$. Then, the proof follows in a similar way as Statement 3 of Theorem \ref{col:feasibility_stability_1}. Dropping time dependence when inside the zero-decrease set, we have $\hat{\mathcal{O}}=\hat{M}(\mathbf{z}^{\bullet},\mathbf{v}^{\bullet};\hat{\mathbf{q}})=M(\mathbf{z}^{\bullet},\mathbf{v}^{\bullet})$ since the linearization trajectory coincides with the PFIT. Assuming $(\Delta \mathbf{z}^{\bullet},\Delta \mathbf{v}^{\bullet})=(\mathbf{z}^{\circ},\mathbf{v}^{\circ})-(\mathbf{z}^{\bullet},\mathbf{v}^{\bullet})\neq 0$ and denoting $\mathbf{a}=M(\mathbf{z}^{\bullet},\mathbf{v}^{\bullet})-\mathcal{O}^{\circ}$, the inequality $M(\tilde{\mathbf{z}},\tilde{\mathbf{v}}) \leq M(\mathbf{z}^{\bullet},\mathbf{v}^{\bullet})-\lambda \mathbf{a}$ follows for all $\lambda \in [0,1]$, where $(\tilde{\mathbf{z}},\tilde{\mathbf{v}})$ is defined as in Theorem \ref{col:feasibility_stability_1}. With the same RFIT candidate as Theorem \ref{col:feasibility_stability_1}, there exists some $\bar{\lambda} \in (0,1]$ such that
\begin{align}
\label{eq:new_ineq_222}
    \hat{\mathcal{O}} &\leq \hat{M}(\tilde{\mathbf{z}},\tilde{\mathbf{v}};\hat{\mathbf{q}})+\lambda^2 \mathbf{b} \\
    &=M(\tilde{\mathbf{z}},\tilde{\mathbf{v}})+\lambda \nabla M(\tilde{\mathbf{z}},\tilde{\mathbf{v}})^{\top} \begin{bmatrix} \Delta \mathbf{z}^{\bullet} \\ \Delta\mathbf{v}^{\bullet} \end{bmatrix}+\frac{\rho \lambda^2}{2} \norm{ \begin{bmatrix} \Delta \mathbf{z}^{\bullet} \\ \Delta\mathbf{v}^{\bullet}\end{bmatrix}}_{2}^{2}+\lambda^2 \mathbf{b} \nonumber \\
    & \leq \hat{\mathcal{O}}-\lambda \tilde{\mathbf{a}}+\lambda^2 \tilde{\mathbf{b}}, \quad \forall \ \lambda \in [0,\bar{\lambda}], \nonumber
\end{align}
where $\tilde{\mathbf{a}} = \mathbf{a}-\nabla M(\tilde{\mathbf{z}},\tilde{\mathbf{v}})^{\top}(\Delta \mathbf{z}^{\bullet},\Delta \mathbf{v}^{\bullet})$ and $\tilde{\mathbf{b}}>0$ consists of the coefficients of $\lambda^2$ in the second line of \eqref{eq:new_ineq_222}. Under strict convexity of $M$ in Assumption \ref{ass:convex}, it follows that $\tilde{\mathbf{a}}>0$ for all $\lambda \in [0,1]$ if $(\Delta \mathbf{z}^{\bullet},\Delta \mathbf{v}^{\bullet}) \neq 0$. Then, \eqref{eq:new_ineq_222} is a contradiction since $-\lambda \tilde{\mathbf{a}}+\lambda^2 \tilde{\mathbf{b}}<0$ whenever $\lambda \in (0,\tilde{\mathbf{a}}/\tilde{\mathbf{b}})$. Hence, we shall have $\tilde{\mathbf{a}}=0$, which holds if and only if $\mathbf{a}=0$ under strict convexity. Then, the same arguments of Statement 3 in Theorem \ref{col:feasibility_stability_1} conclude the proof.
\end{proof}

\subsection{Computational aspects}
Problems \eqref{eq:periodic_MPC_implementation_1} and \eqref{eq:periodic_MPC_implementation_2} are defined over $(N+1+T)\cdot(\mathsf{f}+\mathsf{v}n_u)$ number of variables, and each set $\mathbb{S}$ is defined using $\mathsf{e}+\mathsf{v}\cdot(m\mathsf{f}+m_{\mathcal{X}}+m_{\mathcal{U}})$ number of inequality constraints, where $m_{\mathcal{X}}$ and $m_{\mathcal{U}}$ are the number of halfspaces defining $\mathcal{X}$ and $\mathcal{U}$ respectively. The terminal set consists of $m\mathsf{v}\mathsf{f}+T\cdot(\mathsf{e}+\mathsf{v}\cdot(m_{\mathcal{X}}+m_{\mathcal{U}}))$ number of inequalities. The number of variables and constraints scale linearly with $N$ and $T$, along with the number of LDI vertices $m$, polytope vertices $\mathsf{v}$, and halfspaces $\mathsf{f}$ and $\mathsf{e}$. To reduce the complexity in the terminal region, the constraints can be replaced by \eqref{eq:terminal_sufficient} as discussed in Remark \ref{remark:interpretation}. The choice of $\epsilon$, however, must be consistent with Assumption \ref{ass:interiority}. Alternatively, tube parameterizations such as those presented in~\cite{badalamenti2025efficientcctmpc,mulagaleti2026dualmpcquasilinearparameter} can be adopted, which restrict the polytope class with the goal of reducing computational complexity. While the approach of \cite{badalamenti2025efficientcctmpc} requires precomputing optimal PFIT offline, \cite{mulagaleti2026dualmpcquasilinearparameter} overcomes this limitation.

\section{Numerical Examples} 
\label{sec:Sect5} 
This section presents two numerical examples. In both cases, \(\gamma=0.95\). All simulations are performed in MATLAB on a laptop equipped with an Intel i5-8350U CPU.

\subsection{Illustrative example} \label{subsec:2D_integrator} 
We consider the uncertain double integrator
\begin{align}
\label{eq:2DIntegrator}
    x^+ = \begin{bmatrix}
        1 + \zeta & 1 + \zeta \\ 0 & 1 + \zeta
    \end{bmatrix} x + \begin{bmatrix}
        0 \\ 1 + \zeta
    \end{bmatrix}u + w, \quad | \zeta | \leq 0.05.
\end{align}
The LDI in~\eqref{eq:system} is obtained from the two extremal realizations of \(\zeta\), with $m=2$. The state and input constraints are $\mathcal{X}= [-8,3]\times[-1.5,4]$ and $\mathcal{U}=[-1,1]$, the additive disturbance satisfies $\mathcal{W}=\{w:|w|\leq [0\ \ 0.1]^\top\}$. The matrix $F$ in~\eqref{eq:PolytopicRFIT} is constructed as $F_i=[\sin(\frac{i-1}{10}2\pi)\ \ \cos(\frac{i-1}{10}2\pi)]$, $i=1,\ldots,10$, resulting in $\fe=\ve=\ee=10$. The control scheme uses $N = 2$ and $T = 6$. Denoting the \textit{minimal configuration-constrained RCI set} as $X(z^{\mathrm{RCI}})$ obtained using the optimizer $(z^{\mathrm{RCI}},v^{\mathrm{RCI}})$ of the QP
\begin{align}
\label{eq:optimalRCI}
\min_{z,v} &\ 10^2\norm{z}_2^2+\norm{v}_2^2 \ \text{s.t.} \ (z,v,z)\in\mathbb{S},
\end{align}
we define the periodic stage cost in Problem~\eqref{eq:DRTO_problem} as
\begin{align*}\ell_{[t+l]}(z_l,v_l) \coloneqq 10^2\norm{z_l-(z^{\mathrm{RCI}}+F x^{\mathrm{ref}}_{[t+l]})}_2^2 +
\norm{v_l}_2^2,
\end{align*}
where $\mathbf{x}^{\mathrm{ref}} = (x^{\mathrm{ref}}_0,\ldots,x^{\mathrm{ref}}_{T-1})$ is a given reference in the state space. This choice of cost results in the computation of an optimal PFIT that tracks the periodic sequence of sets $(X(z^{\mathrm{RCI}}+Fx_{[t]}^{\mathrm{ref}}),t \in \mathbb{N})$ while respecting the LDI dynamics. 

The tracking cost matrices in the MPC formulation are $Q=Q_\mathrm{v}+10^{-3}\I_{20}$, where $Q_\mathrm{v} = \sum_{j=1}^{\ve}(\bar{V}-V_j,\bar{U}-U_j)^\top \operatorname{blkd(10\I_2,1)}(\bar{V}-V_j,\bar{U}-U_j)$, and $P=(1-\gamma^2)^{-1}Q$. Since the resulting cost in Problem~\eqref{eq:periodic_MPC_implementation_1} is quadratic, the approximation in~\eqref{eq:quadratic_upper_bound} is not required. 

Figure~\ref{fig:2D_statespace} illustrates the closed-loop evolution of the tube and reference sets inside the feasible region denoted in yellow: blue sets represent the system tube, green sets denote the artificial PFIT, and red sets indicate the desired optimal PFITs corresponding to different reference vectors \(\mathbf{x}^{ref}\). These results confirm that the proposed control scheme effectively drives the system to follow the desired periodic tube. Figure~\ref{fig:LyapunovExamples} shows the evolution of the Lyapunov function $\mathcal{L}_t$ defined in~\eqref{eq:Lyapunov_function}, which converges to zero after each change in the reference.
\begin{figure}
    \centering
    \includegraphics[width=1\linewidth]{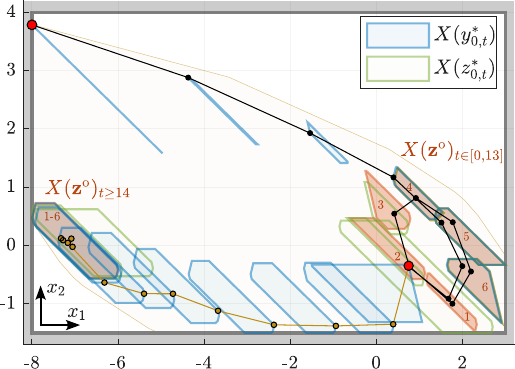}
    \caption{Closed-loop evolution of system~\eqref{eq:2DIntegrator}. Dotted lines indicate state trajectories, with red dots marking states at $t=0$ and at the reference change. For $t\in[0,13]$, $x^{ref}_{[t+l]} = (3\cos(\theta_{[t+l]})+1.5,\,\sin(\theta_{[t+l]}))$, with $\theta_{[t+l]}=-(\frac{[t+l]}{T}2\pi)-\frac{\pi}{2}$. When $t\ge14$, $x^{ref}_{[t+l]}=(-10,0)$. Sets $X(y_0^{\star})$ and $X(z_0^{\star})$ are shown for $t\leq5$ and $t\ge14$.}
    \label{fig:2D_statespace}
\end{figure}

\subsection{Ball--plate system}
\label{subsec:ball_plate}

We consider a linearized ball--plate system from~\cite{LIMON_Periodic_7172461}, discretized with sampling time \(0.4\,\mathrm{s}\). The state is \(x=(x^1,x^2)\), where \(x^i=(x_b^i,\dot x_b^i,\theta_p^i,\dot\theta_p^i)\), and \(u=(u^1,u^2)\) contains the plate angular accelerations. The nominal single-axis state and input matrices are 
\begin{align} 
\bar A &= \begin{bmatrix} 1 & 0.4 & 0.0704 & 8\cdot10^{-4}\\ 0 & 1 & 2.8 & 0.0704\\ 0 & 0 & 1 & 0.4\\ 0 & 0 & 0 & 1 \end{bmatrix}, 
& \bar B &= \begin{bmatrix} 0\\ 8\cdot10^{-4}\\ 0.0104\\ 0.4 \end{bmatrix}. 
\end{align} 
To account for uncertainty in the plate-to-ball coupling, define the single-axis vertices \((\bar A_1,\bar B_1)\) and \((\bar A_2,\bar B_2)\) by scaling the corresponding coupling terms of \(\bar A\) and \(\bar B\) by \(1-\rho_{\mathrm{dyn}}\) and \(1+\rho_{\mathrm{dyn}}\), respectively, where \(\rho_{\mathrm{dyn}}=0.01\). The polytopic LDI is thus defined by the two vertex pairs 
\begin{align} A_i &= \operatorname{blkd}(\bar A_i,\bar A_i), & B_i &= \operatorname{blkd}(\bar B_i,\bar B_i), \quad i\in\{1,2\}. \label{eq:BallPlateSysEq} 
\end{align}
State constraints reflect the physical limitations of the plate
$$\mathcal{X}=\left\{ |x^{1,2}|\leq (0.25,10,\frac{\pi}{3},20),\ |x_b^1|+|x_b^2|\leq0.3\right\}$$
with an additional constraint \( x_b^1\le 0.15 \) to exclude an undesired region. The input set is $\mathcal{U}=\{ u\mid\norm{u}_\infty\leq 3\, \mathrm{rad/s^2}\}$, while additive disturbance affects the ball position and plate-angle coordinates through 
\begin{align*}
    \mathcal{W} = \scalemath{0.95}{\left\{ \operatorname{blkd}(B_w,B_w) w \mid \norm{w}_\infty \leq 5\cdot10^{-5}\right\}},\ B_w= \left[\scalemath{0.9}{e_4(1)\ e_4(3)}\right].
\end{align*}

The matrix \(F\) in~\eqref{eq:PolytopicRFIT} is obtained from the Cartesian product of a regular hexagon in the \((x_b^1,x_b^2)\)-plane and a six-dimensional simplex in the remaining coordinates. The seed offset is \( y_\sigma= ( \mathbf 1_7,\frac{\sqrt{3}}{2}\mathbf 1_6) \), and the template is transformed offline using the procedure in~\cite[Sec.~4.2]{badalamenti2025efficientcctmpc}. The resulting CC-polytope is simple and has $\fe=13$, $\ve=42$, $\ee=7$. 

The control parameters are \(N=2\) and \(T=16\). The periodic stage cost $\ell_{[t+l]}(z_l,v_l)$ in Problem~\eqref{eq:DRTO_problem} is selected as
\begin{multline} 
\label{eq:ball_plate_stage_cost} 
 \sum_{j=1}^{\ve}(10\norm{CV_j z_l-r_{[t+l]}}_2^2 +0.1\sum_{k=1}^{n_u}\log\cosh(2v_{jk,l})), 
\end{multline}
where $\scalemath{0.95}{C= \begin{bmatrix} e_8(1)&e_8(5)\end{bmatrix}^\top}$, $\scalemath{0.95}{r_{[t+l]}= ( 0.25\cos\theta_{[t+l]}, 0.2\sin\theta_{[t+l]})}$ and $\theta_{[t+l]}=\pi-2\pi[t+l]/T$.
Here, \(v_{jk,l}\) denotes the \(k\)-th component of the vertex input associated with vertex \(j\) at periodic stage \(l\).

For the QP implementation, the scalar in Assumption~\ref{ass:quadratic_majorizer} is replaced by \(H = \mathrm{blkd}(H_l,\ldots,H_l)\), where \(H_l = 2\nabla_{(z,v)}^2\ell_{[t+l]}(0,0)\). For~\eqref{eq:ball_plate_stage_cost}, \(H_l\) is a common upper bound for all periodic stages. The repeated block structure implies cyclic invariance of \(H\), so Corollary~\ref{cor:QP_properties} remains valid with \(\rho\I\) replaced by \(H\). The tracking weights are \(Q=\mathrm{blkd}(\I_{13},10^{-3}\I_{84})\), \(P=(1-\gamma^2)^{-1}Q\).

\begin{figure} 
\centering
\includegraphics[width=\linewidth]{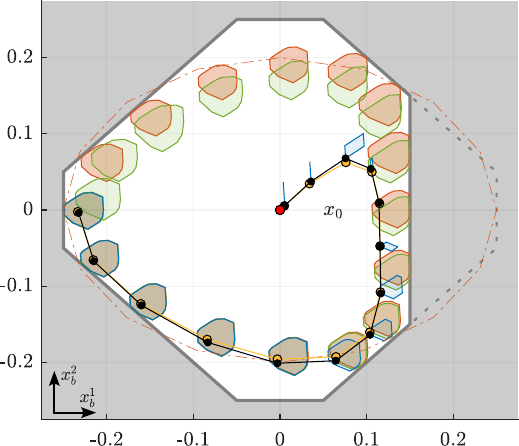}
\caption{Closed-loop evolution of the ball--plate system under Problem~\eqref{eq:periodic_MPC_implementation_1} (yellow) and Problem~\eqref{eq:periodic_MPC_implementation_2} (black), projected onto the \((x_b^1,x_b^2)\)-plane. Blue, green, and red sets denote the projected RFIT, artificial-PFIT, and optimal-PFIT cross-sections, respectively. The dashed red curve denotes the periodic reference.} 
\label{fig:BallPlate_statespace} 
\end{figure}

Figure~\ref{fig:BallPlate_statespace} compares the closed-loop evolutions obtained from Problems~\eqref{eq:periodic_MPC_implementation_1} and~\eqref{eq:periodic_MPC_implementation_2}. Both controllers approach the same optimal PFIT, while the QP formulation exhibits a slower transient. This behavior is also visible in Figure~\ref{fig:LyapunovExamples}. 
The average main-OCP solution times for Problems~\eqref{eq:periodic_MPC_implementation_1} and \eqref{eq:periodic_MPC_implementation_2} were \(54.34\,\mathrm{s}\) and \(54.82\,\mathrm{s}\), respectively, using IPOPT through CasADi~\cite{Andersson2019}. Using Gurobi through CasADi~\cite{gurobi} for the latter reduced the average solution time to \(2.82\,\mathrm{s}\), showing that its computational advantage is realized by exploiting the QP structure with a specialized solver.
\begin{figure}
\centering 
\includegraphics[width=0.49\linewidth]{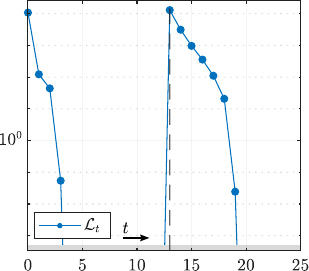} \hfill \includegraphics[width=0.49\linewidth]{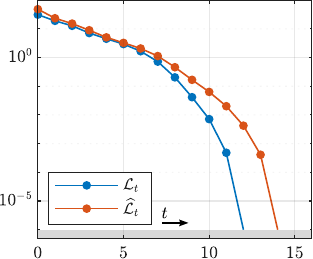} 
\caption{Lyapunov functions for system~\eqref{eq:2DIntegrator} (left) and system~\eqref{eq:BallPlateSysEq} (right). The vertical dashed line indicates a change in the reference, while the bottom gray areas denote numerical tolerance limits.} 
\label{fig:LyapunovExamples} 
\end{figure}

\section{Conclusions} 
\label{sec:Sect6}
This paper extends the CCTMPC framework to robust periodic operation of systems described by polytopic LDIs.
The proposed single-layer formulation jointly optimizes an RFIT and an artificial PFIT, guaranteeing recursive feasibility and robust constraint satisfaction under online changes of the periodic cost, and convergence to the optimal PFIT for a fixed periodic cost structure. The terminal ingredients are imposed through linear admissibility and contractivity conditions relative to the artificial PFIT, without constructing an explicit target-dependent terminal set.
For nonquadratic periodic costs, a quadratic upper model yields a QP formulation preserving the established closed-loop properties.
The numerical examples demonstrate robust periodic operation under additive and parametric uncertainty, with a reduction of more than one order of magnitude in the main-OCP solution time when the QP formulation is solved using a specialized solver.
Since the online problem size also scales with the PFIT period, reduced artificial-tube parameterizations constitute a relevant direction for further computational improvements.

\balance
\bibliographystyle{IEEEtran}
\bibliography{normalized_bibliography}

\end{document}